\documentclass[american,aps,pra,reprint,superscriptaddress]{revtex4-1}
\usepackage[latin9]{inputenc}
\usepackage{babel}
\usepackage{amsmath}
\usepackage{amsthm}
\usepackage{amssymb}
\usepackage{graphicx}
\usepackage[unicode=true,pdfusetitle,
 bookmarks=true,bookmarksnumbered=false,bookmarksopen=false,
 breaklinks=false,pdfborder={0 0 0},backref=false,colorlinks=false]
 {hyperref}
\hypersetup{
 colorlinks=true,allcolors=magenta}

\makeatletter
\theoremstyle{plain}
\newtheorem{thm}{\protect\theoremname}
  \theoremstyle{plain}
  \newtheorem{lem}[thm]{\protect\lemmaname}
\newenvironment{lyxlist}[1]
{\begin{list}{}
{\settowidth{\labelwidth}{#1}
 \setlength{\leftmargin}{\labelwidth}
 \addtolength{\leftmargin}{\labelsep}
 }}
{\end{list}}

\usepackage{babel}
\usepackage{txfonts}
\usepackage{braket}

\newcommand{\abs}[1]{\left| #1\right|}

\newcommand{\proj}[1]{\ketbra{#1}{#1}}

\newcommand{\Tr}{\text{Tr\,}}

\newcommand{\ketbra}[2]{\ensuremath{| #1 \rangle \!\langle #2 |}}

\newcommand{\dist}{D}

\newcommand{\pur}{\mathcal{P}}
\newcommand{\coh}{\mathcal{C}}
\newcommand{\dis}{\mathcal{D}}
\newcommand{\ent}{\mathcal{E}}

\providecommand{\lemmaname}{Lemma}
\providecommand{\theoremname}{Theorem}

\makeatother

  \providecommand{\lemmaname}{Lemma}
\providecommand{\theoremname}{Theorem}

\begin{document}

\title{Maximal Coherence and the Resource Theory of Purity}

\author{Alexander Streltsov}

\email{streltsov.physics@gmail.com}

\affiliation{Faculty of Applied Physics and Mathematics, Gda\'{n}sk University
of Technology, 80-233 Gda\'{n}sk, Poland}

\affiliation{National Quantum Information Centre in Gda\'{n}sk, 81-824 Sopot,
Poland}

\affiliation{Dahlem Center for Complex Quantum Systems, Freie Universität Berlin,
D-14195 Berlin, Germany}

\author{Hermann Kampermann}

\affiliation{Institut für Theoretische Physik III, Heinrich-Heine-Universität
Düsseldorf, D-40225 Düsseldorf, Germany}

\author{Sabine Wölk}

\affiliation{Institute for Theoretical Physics, University of Innsbruck, Technikerstraße
25, A-6020 Innsbruck, Austria}

\affiliation{Naturwissenschaftlich-Technische Fakultät, Universität Siegen, Walter-Flex-Str.
3, D-57072 Siegen, Germany}

\author{Manuel Gessner}

\affiliation{QSTAR, INO-CNR and LENS, Largo Enrico Fermi 2, I-50125 Firenze, Italy}

\affiliation{Istituto Nazionale di Ricerca Metrologica, Strada delle Cacce 91,
I-10135 Torino, Italy}

\author{Dagmar Bruß}

\affiliation{Institut für Theoretische Physik III, Heinrich-Heine-Universität
Düsseldorf, D-40225 Düsseldorf, Germany}
\begin{abstract}
The resource theory of quantum coherence studies the off-diagonal
elements of a density matrix in a distinguished basis, whereas the
resource theory of purity studies all deviations from the maximally
mixed state. We establish a direct connection between the two resource
theories, by identifying purity as the maximal coherence which is
achievable by unitary operations. The states that saturate this maximum
identify a universal family of maximally coherent mixed states. These
states are optimal resources under maximally incoherent operations,
and thus independent of the way coherence is quantified. For all distance-based
coherence quantifiers the maximal coherence can be evaluated exactly,
and is shown to coincide with the corresponding distance-based purity
quantifier. We further show that purity bounds the maximal amount
of entanglement and discord that can be generated by unitary operations,
thus demonstrating that purity is the most elementary resource for
quantum information processing.
\end{abstract}
\maketitle

\section{Introduction}

A number of different quantum features are considered as important
resources for applications of quantum information theory. Entanglement~\cite{PhysRevLett.78.2275,Bruss2002,Plenio2007,Horodecki2009},
quantum discord~\cite{Ollivier2002,Henderson2001,RevModPhys.84.1655,Streltsov2014,Adesso2016,Adesso2016b},
and quantum coherence~\cite{Aberg2006,Plenio2014,Levi2014,WinterResourceCoherence,Marvian2016,CoherenceResource}
have been identified as necessary ingredients for the successful implementation
of tasks, such as quantum cryptography~\cite{RevModPhys.74.145},
quantum algorithms \cite{NielsenChuang,RevModPhys.74.145} and quantum
metrology \cite{Giovannetti2011,Demkowicz2012,Varenna,Toth2014,PhysRevLett.112.210401}.
Quantum resources can be formally classified in the framework of resource
theories~\cite{Horodecki2013,SpekkensResourceTheory}, where the
state space is divided into free states and resource states. Moreover,
a set of free operations, which cannot turn a free state into a resource
state, is identified~\footnote{Throughout this paper, a quantum operation (or just ``operation'')
is used as a synonym for a completely positive trace-preserving map.}. The possibility of conversion between two resource states via free
operations is a central issue within a resource theory, as it introduces
a natural order of the resource states. A suitable measure for the
resource must be non-increasing under free operations. Equipped with
suitable measures, one is able to quantify the resource in any given
quantum state.

States that maximize such measures are called extremal resource states~\footnote{Here, the term ``extremal'' is not synonymous to extremal elements
of convex sets.}. Every quantum state can then be characterized by the minimal rate
of extremal resource states needed to create it (resource cost), or
the maximal rate for creating an extremal resource state from it (distillable
resource), using the free operations~\footnote{This has to be understood in the asymptotic setting, where ``rate''
means the asymptotic fraction of required (distilled) resource states
per copy of the desired (given) quantum state $\rho$. Widely used
examples for such asymptotic rates are entanglement cost and distillable
entanglement, we refer to Ref.~\cite{Plenio2007} for their formal
definition.}. A number of different resource theories have been developed in the
context of quantum information theory~\cite{Horodecki2013,SpekkensResourceTheory},
prominent examples being entanglement~\cite{PhysRevLett.78.2275,Bruss2002,Plenio2007,Horodecki2009}
and coherence~\cite{Aberg2006,Plenio2014,Levi2014,WinterResourceCoherence,Marvian2016,CoherenceResource}.

While the concept of coherence is basis-dependent by its very definition,
both entanglement and quantum discord are locally basis-independent.
However, entanglement and discord usually change if a global unitary
is applied. It is clear, however, that the unitary activation of these
resources must be limited in terms of some basis-independent quantity
of the initial quantum state. As we will show in rigorous quantitative
terms, this fundamental quantity is identified as purity. Specifically,
we show how purity can be used to establish quantum coherence by a
unitary operation. This further provides direct bounds on the amount
of entanglement and discord that can be reached by unitary operations,
since these quantities can be traced back to coherences in a specific
many-body basis. These results hold for all suitable distance-based
quantifiers.

A resource theory of purity was introduced in~\cite{Horodecki2003}
for the asymptotic limit of infinitely many copies of the quantum
state. The finite-copy scenario was considered more recently~\cite{GourResThOpPhysRep15}.
Our results relate both of these approaches directly to the resource
theory of coherence. In general, purity can be interpreted as the
maximal coherence, maximized over all unitaries. Depending on the
chosen coherence monotone, we recover either the asymptotic or the
finite-copy resource theory of purity, by maximizing over unitary
operations, or even more generally, over all unital operations. As
one of our main results, we are able to identify the states that maximize
any given coherence monotone for a fixed spectrum of the density matrix.
These states define a universal set of maximally coherent mixed states.
The coherence of these states can be evaluated exactly for any distance-based
coherence monotone, and is shown to coincide with its distance-based
purity.

\section{Resource theory of quantum coherence}

In the following, we recall the resource theory of coherence~\cite{Aberg2006,Plenio2014,Levi2014,WinterResourceCoherence,Marvian2016,CoherenceResource}
and then identify the family of maximally coherent mixed states. The
free states of this resource theory are called \emph{incoherent states},
these are states which are diagonal in a fixed basis $\{\ket{i}\}$,
i.e., 
\begin{equation}
\sigma=\sum_{i}p_{i}\ket{i}\!\bra{i}.
\end{equation}
The set of all incoherent states will be denoted by $\mathcal{I}$.
The definition of free operations is not unique, and several approaches
have been presented in the literature~\cite{Marvian2016,CoherenceResource}. 

The historically first and most general approach was suggested in~\cite{Aberg2006},
where the set of \emph{maximally incoherent operations} (MIO) was
considered. These are all operations which cannot create coherence,
i.e., 
\begin{equation}
\Lambda_{\mathrm{MIO}}[\sigma]\in\mathcal{I}
\end{equation}
 for any incoherent state $\sigma\in\mathcal{I}$. Another important
family is the set of \textit{incoherent operations} (IO)~\cite{Plenio2014}.
These are operations which admit a Kraus decomposition 
\begin{equation}
\Lambda_{\mathrm{IO}}[\rho]=\sum_{i}K_{i}\rho K_{i}^{\dagger}
\end{equation}
 with incoherent Kraus operators $K_{i}$, i.e., $K_{i}\ket{m}\sim\ket{n}$,
where the states $\ket{m}$ and $\ket{n}$ belong to the incoherent
basis. We also note that IO is a strict subset of MIO~\cite{WinterResourceCoherence,Chitambar2016,Chitambar2016b}
\begin{equation}
\mathrm{IO}\subset\mathrm{MIO},\label{eq:IOMIO}
\end{equation}
and the inclusion is strict even for single-qubit states~\cite{ChitambarPhysRevA.95.019902}.
While we will focus on the sets MIO and IO in this work, other relevant
sets of operations have been discussed in recent literature, based
on physical or mathematical considerations~\cite{Levi2014,WinterResourceCoherence,Marvian2016,Yadin2016,Chitambar2016,Chitambar2016b,ChitambarPhysRevA.95.019902,Marvian2016b,GenuineCoherence}.
An extension of quantum coherence to multipartite systems has also
been presented~\cite{Bromley2015,Streltsov2015}, which made it possible
to investigate the resource theory of coherence in distributed scenarios~\cite{Chitambar2016c,Ma2016,Streltsov2016,Chitambar2016d,Matera2016,Streltsov2015b}.
A review over alternative frameworks of coherence and their interpretation
can be found in~\cite{CoherenceResource}.

The amount of coherence in a given state can be quantified via \emph{coherence
monotones}. These are nonnegative functions $\coh$ which do not increase
under the corresponding set of free operations, i.e., for a MIO monotone
we have $\coh(\Lambda_{\mathrm{MIO}}[\rho])\leq\coh(\rho)$. Since
MIO is the most general set of free operations for any resource theory
of coherence, a MIO monotone is also a monotone in any other coherence
theory. An important example are distance-based coherence monotones:
\begin{align}
\coh(\rho)=\inf_{\sigma\in\mathcal{I}}D(\rho,\sigma),\label{eq:C}
\end{align}
where $D$ is a suitable distance on the space of quantum states.
Such quantifiers were studied in~\cite{Aberg2006,Plenio2014}, the
most prominent example being the relative entropy of coherence 
\begin{equation}
\coh_{\mathrm{r}}(\rho)=\min_{\sigma\in\mathcal{I}}S(\rho||\sigma)
\end{equation}
with the quantum relative entropy $S(\rho||\sigma)=\mathrm{Tr}[\rho\log_{2}\rho]-\mathrm{Tr}[\rho\log_{2}\sigma]$~\footnote{We note that the quantum relative entropy is not a distance in the
mathematical sense, as it is not symmetric and does not fulfill the
triangle inequality.}. Remarkably, this quantity admits a closed expression~\cite{Plenio2014}
and coincides with the distillable coherence under MIO and IO and
also with the coherence cost under MIO~\cite{WinterResourceCoherence}:
\begin{equation}
\coh_{\mathrm{r}}(\rho)=S(\Delta[\rho])-S(\rho).
\end{equation}
Here, $S(\rho)=-\mathrm{Tr}[\rho\log_{2}\rho]$ is the von Neumann
entropy and $\Delta[\rho]=\sum_{i}\braket{i|\rho|i}\ket{i}\!\bra{i}$
denotes dephasing in the incoherent basis.

For a general distance-based coherence quantifier as given in Eq.~(\ref{eq:C})
one usually considers nonnegative distances $D$ which are contractive
under any quantum operation $\Lambda$: 
\begin{align}
D(\Lambda[\rho],\Lambda[\sigma])\leq D(\rho,\sigma).\label{eq:contractivity}
\end{align}
Any such distance gives rise to a MIO monotone~\cite{Plenio2014,CoherenceResource}.
Examples for such distances are the relative Rényi entropy 
\begin{equation}
D_{\alpha}(\rho||\sigma)=\frac{1}{\alpha-1}\log_{2}\mathrm{Tr}[\rho^{\alpha}\sigma^{1-\alpha}]
\end{equation}
 and the quantum relative Rényi entropy 
\begin{equation}
D_{\alpha}^{\mathrm{q}}(\rho||\sigma)=\frac{1}{\alpha-1}\log_{2}\mathrm{Tr}[(\sigma^{\frac{1-\alpha}{2\alpha}}\rho\sigma^{\frac{1-\alpha}{2\alpha}})^{\alpha}].
\end{equation}
While $D_{\alpha}$ is contractive for $\alpha\in[0,2]$, the function
$D_{\alpha}^{\mathrm{q}}$ is contractive in the range $\alpha\in[\frac{1}{2},\infty]$~\cite{QuantRenyEntrGeneralTomamichel,Leditzky2016}.
We can now define a family of coherence monotones in the following
way: 
\begin{equation}
\coh_{\alpha}(\rho)=\begin{cases}
\inf_{\sigma\in\mathcal{I}}D_{\alpha}(\rho||\sigma) & \mathrm{for}\,\,0<\alpha<1,\\
\inf_{\sigma\in\mathcal{I}}D_{\alpha}^{\mathrm{q}}(\rho||\sigma) & \mathrm{for}\,\,\alpha>1.
\end{cases}\label{eq:Calpha}
\end{equation}
This quantity is a MIO monotone in the range $\alpha\in[0,\infty]$.
In the limit $\alpha\rightarrow1$ both functions $D_{\alpha}(\rho||\sigma)$
and $D_{\alpha}^{\mathrm{q}}(\rho||\sigma)$ coincide with the relative
entropy $S(\rho||\sigma)$. Coherence quantifiers of this type were
studied in~\cite{Chitambar2016,Chitambar2016b,RenyiEntrCohMeas16}.
A related approach based on Tsallis relative entropies has also been
investigated~\cite{Rastegin2016}.

Several MIO monotones have additional desirable properties such as
strong monotonicity under IO and convexity~\cite{Plenio2014,CoherenceResource}.
This is in particular the case for the relative entropy of coherence~\cite{Plenio2014}.
While any MIO monotone is also an IO monotone, the other direction
is less clear. In particular, the $l_{1}$-norm of coherence 
\begin{equation}
\coh_{l_{1}}(\rho)=\sum_{i\neq j}|\rho_{ij}|
\end{equation}
 is known to be an IO monotone~\cite{Plenio2014}, but violates monotonicity
under MIO~\cite{Bu2016}. Another IO monotone which is not a MIO
monotone is the coherence of formation 
\begin{equation}
\coh_{\mathrm{f}}(\rho)=\min\sum_{i}p_{i}S(\Delta[\ket{\psi_{i}}\!\bra{\psi_{i}}]),
\end{equation}
where the minimum is taken over all pure state decompositions $\{p_{i},\ket{\psi_{i}}\}$
of the state $\rho$ \cite{Yuan2015,WinterResourceCoherence,Hu2016}.

We also note that coherence of formation is equal to coherence cost
under IO~\cite{WinterResourceCoherence}, and $l_{1}$-norm of coherence
is related to the path information in multi-path interferometer~\cite{Bagan2016,Bera2015}.

\section{Maximally coherent mixed states}

Since coherence is a basis-dependent concept, a unitary operation
will in general change the amount of coherence in a given state. In
the following, we will focus on the question: which unitary maximizes
the coherence of a given state $\rho$? The corresponding figure of
merit is given as follows: 
\begin{align}
\coh_{\max}(\rho):=\sup_{U}\coh(U\rho U^{\dagger}).\label{eq:Cmax}
\end{align}
If the supremum in Eq.~(\ref{eq:Cmax}) is realized for the unitary
$V$, the corresponding state $\rho_{\max}=V\rho V^{\dagger}$ will
be called \emph{maximally coherent mixed state}. This definition is
in full analogy to maximally entangled mixed states investigated in~\cite{PhysRevA.62.022310,MaxEntMixed2QubitStatesPRA01,PhysRevA.64.030302,PhysRevA.67.022110}.
Maximally coherent mixed states were first introduced for specific
measures of coherence in~\cite{Singh2015}, and studied further more
recently in~\cite{Yao2016}.

While the relative entropy of coherence admits a closed formula, the
evaluation of general coherence monotones is considered as a hard
problem~\cite{CoherenceResource}. It is thus reasonable to believe
that the maximization in Eq.~(\ref{eq:Cmax}) is out of reach. Quite
surprisingly, we will now show that the supremum in Eq.~(\ref{eq:Cmax})
can be evaluated in a large number of relevant scenarios. In particular,
we will see that there exists a \emph{universal} maximally coherent
mixed state, which does not depend on the particular choice of coherence
monotone. These results will also lead us to a closed expression of
$\coh_{\max}$ for all distance-based coherence monotones. 
\begin{thm}
\label{thm:MCM} Among all states $\rho$ with a fixed spectrum $\{p_{n}\}$,
the state 
\begin{align}
\rho_{\max}=\sum_{n=1}^{d}p_{n}\ket{n_{+}}\!\bra{n_{+}},\label{eq:MCM}
\end{align}
is a maximally coherent mixed state with respect to any MIO monotone.
Here, $\{\ket{n_{+}}\}$ denotes a mutually unbiased basis with respect
to the incoherent basis $\{\ket{i}\}$, i.e., $|\!\braket{i|n_{+}}\!|^{2}=\frac{1}{d}$,
where $d$ is the dimension of the Hilbert space. \end{thm}
\begin{proof}
\noindent We will actually prove an even stronger statement. In particular,
we will show that for any unitary $U$, the transformation $\rho_{\max}\rightarrow U\rho_{\max}U^{\dagger}$
can be achieved via MIO, i.e., 
\begin{align}
\Lambda_{\mathrm{MIO}}[\rho_{\max}]=U\rho_{\max}U^{\dagger}.
\end{align}
The proof of the theorem then follows by using monotonicity of $\coh$
under MIO: 
\begin{equation}
\coh(U\rho_{\max}U^{\dagger})=\coh(\Lambda_{\mathrm{MIO}}[\rho_{\max}])\leq\coh(\rho_{\max}).
\end{equation}
The operation $\Lambda_{\mathrm{MIO}}$ which achieves this transformation
has Kraus operators $K_{n}=U\ket{n_{+}}\bra{n_{+}}$. Note that bases
$\{\ket{n_{+}}\}$ and $\{\ket{i}\}$ are mutually unbiased, which
implies that $\sum_{n}K_{n}\sigma K_{n}^{\dagger}=\openone/d$ for
any incoherent state $\sigma$. This means that the operation $\Lambda_{\mathrm{MIO}}[\rho]=\sum_{n}K_{n}\rho K_{n}^{\dagger}$
is indeed maximally incoherent. In the final step, note that $\sum_{n}K_{n}\rho_{\max}K_{n}^{\dagger}=U\rho_{\max}U^{\dagger}$,
and the proof is complete.
\end{proof}
This theorem has several important implications. First, it implies
that the state $\rho_{\max}$ is a resource with respect to all states
with the same spectrum. Second, this theorem provides an alternative
simple proof for the fact that $l_{1}$-norm of coherence can increase
under MIO~\cite{Bu2016}. This can be seen by combining Theorem~\ref{thm:MCM}
with the fact that the state in Eq.~(\ref{eq:MCM}) is not a maximally
coherent mixed state for the $l_{1}$-norm of coherence~\cite{Yao2016}.
Moreover, a unitary $V$ for an arbitrary state $\rho$ which achieves
the supremum in Eq.~(\ref{eq:Cmax}) for any MIO monotone is given
by $V=\sum_{n=1}^{d}\ket{n_{+}}\!\bra{\psi_{n}}$, where $\{\ket{\psi_{n}}\}$
are the eigenstates of $\rho$.

We will now go one step further and give an explicit expression for
$\coh_{\max}$ for any distance-based coherence monotone. 
\begin{thm}
\label{thm:Cmax}For any distance-based coherence monotone as given
in Eq.~(\ref{eq:C}) with a contractive distance $D$ the following
equality holds:\emph{ 
\begin{equation}
\coh_{\max}(\rho)=\coh(\rho_{\max})=D\left(\rho,\openone/d\right).
\end{equation}
}
\end{thm}
\noindent We refer to Appendix~\ref{sec:Proof-2} for the proof.
Note that Theorem~\ref{thm:Cmax} also holds for all coherence quantifiers
\begin{equation}
\coh_{p}=\min_{\sigma\in\mathcal{I}}||\rho-\sigma||_{p}
\end{equation}
 based on Schatten $p$-norms $||M||_{p}=(\mathrm{Tr}[(M^{\dagger}M)^{p/2}])^{1/p}$
for all $p\geq1$. Equipped with these results, we will show below
in this paper that the resource theory of coherence is closely related
to the resource theory of purity. Before we present these results,
we review the main properties of the resource theory of purity in
the following.

\section{\label{sec:ResourceTheories}Resource theory of purity}

We will now review resource theories of purity based on different
sets of free operations. The discussion summarizes results previously
presented in \cite{GourResThOpPhysRep15}. There exists a hierarchy
of quantum operations which generalize classical bistochastic (purity
non-increasing) maps. We distinguish three types of quantum operations: 
\begin{itemize}
\item Mixture of unitary operations: 
\begin{align}
\Lambda_{\mathrm{MU}}[\rho]=\sum_{i}p_{i}U_{i}\rho U_{i}^{\dag},
\end{align}
with $p_{i}\geq0$, $\sum_{i}p_{i}=1$, and unitary operations $U_{i}$.
\item Noisy operations: 
\begin{align}
\Lambda_{\mathrm{NO}}[\rho]=\Tr_{E}\left[U\left(\rho\otimes\openone_{E}/d\right)U^{\dag}\right],
\end{align}
with a unitary operation $U$. 
\item Unital operations: 
\begin{align}
\Lambda_{\mathrm{U}}[\openone/d]=\openone/d,
\end{align}
i.e., operations which preserve the maximally mixed state.
\end{itemize}
Note that in contrast to the discussion in~\cite{GourResThOpPhysRep15},
we only consider operations which preserve the dimension of the Hilbert
space. It turns out that these operations form a subset hierarchy
\begin{align}
\left\{ \Lambda_{\mathrm{MU}}\right\} \subset\left\{ \Lambda_{\mathrm{NO}}\right\} \subset\left\{ \Lambda_{\mathrm{U}}\right\} ,
\end{align}
see, e.g., Lemma 5 in Ref.~\cite{GourResThOpPhysRep15}.

We call two resource theories equivalent if their respective sets
of free states, as well as their sets of all states coincide, and
additionally if for each $\Lambda_{1}$ with $\Lambda_{1}(\rho)=\sigma$
there exists a $\Lambda_{2}$, such that $\Lambda_{2}(\rho)=\sigma$,
where $\Lambda_{1}$ $(\Lambda_{2})$ is a free operation of resource
theory 1 (2). Due to Lemma 10 in \cite{GourResThOpPhysRep15} the
state conversion abilities are equivalent for the three cases $\Lambda_{\mathrm{MU}}$,
$\Lambda_{\mathrm{NO}}$, and $\Lambda_{\mathrm{U}}$. It follows
that any resource theory which only deviates in the type of operations
as defined above will be equivalent. If not stated otherwise, we will
consider the resource theory of purity based on unital operations
$\Lambda_{\mathrm{U}}$ in the following.

Within the resource theory of purity, the state conversion possibilities
follow from the classical theory of bistochastic maps \cite{UhlmannStochBook,GourResThOpPhysRep15},
using the concept of majorization. A state $\rho$ majorizes another
state $\sigma$, i.e., $\rho\succ\sigma$, if their spectra are in
majorization order: 
\begin{align}
\sum\limits _{i=1}^{k}\lambda_{i}^{\downarrow}(\rho)\geq\sum\limits _{j=1}^{k}\lambda_{j}^{\downarrow}(\sigma)
\end{align}
for all $k\geq1$. Here, $\lambda_{i}^{\downarrow}(\rho)$ denotes
the eigenvalues of $\rho$ in non-increasing order. The aforementioned
relation to the resource theory of purity is established via the following
Lemma. 
\begin{lem}
\label{lem:unital}Given two states $\rho$ and $\sigma$ of the same
dimension, $\rho$ can be converted into $\sigma$ via some unital
operation $\Lambda_{\mathrm{U}}$ if and only if $\rho$ majorizes
$\sigma$: 
\begin{align}
\Lambda_{\mathrm{U}}[\rho]=\sigma\Leftrightarrow\rho\succ\sigma.
\end{align}

\end{lem}
\noindent For the proof of this Lemma we refer to Theorem 4.1.1 in~\cite{NielsenLectureMajorization}
(see also \cite{Uhlmann1970}). Due to the arguments mentioned above,
it follows that the majorization relation is necessary and sufficient
for state conversion via any set of operations presented above.

Fundamental questions in any resource theory address the number of
extremal resource states that can be distilled from a state $\rho$.
In the case of purity this poses the question, how many copies of
a pure single-qubit state $\ket{\psi}_{2}$ can one extract via unital
operations? We will call the corresponding figure of merit \emph{single-shot
distillable purity}. Its formal definition can be given as follows~\footnote{Note that the dimensions $d_{1}$ and $d_{2}$ in Eq.~(\ref{eq:distillablepurity})
are arbitrary finite numbers, up to the requirement that $d\times d_{2}=2^{m}\times d_{1}$.
This guarantees that the unital operation preserves the dimension
of the Hilbert space. As we show in Appendix~\ref{sec:Proof-4},
the optimal choice is $d_{1}=d$ and $d_{2}=2^{\left\lfloor \log_{2}(d/r)\right\rfloor }$,
where $r$ is the rank of $\rho$. This only applies if $\log_{2}(d/r)\geq1$,
as single-shot purity distillation does not work otherwise. By similar
considerations, the optimal choice of dimensions in Eq.~(\ref{eq:puritycost})
is $d_{1}=d$ and d$_{2}=2^{\left\lceil \log_{2}(d\lambda_{\max})\right\rceil }$,
where $\lambda_{\max}$ is the maximal eigenvalue of $\rho$, see
Appendix~\ref{sec:Proof-5} for more details.}: 
\begin{align}
\pur_{\mathrm{d}}^{1}(\rho) & =\max\!\left\{ m:\exists\:\Lambda_{\mathrm{U}}\text{, s.t. }\Lambda_{\mathrm{U}}\!\left[\rho\otimes\frac{\openone}{d_{2}}\right]=\psi_{2}^{\otimes m}\otimes\frac{\openone}{d_{1}}\right\} ,\label{eq:distillablepurity}
\end{align}
where $\Lambda_{\mathrm{U}}$ is a unital operation, $\psi_{2}=\proj{\psi}_{2}$
is a pure single-qubit state, and $\openone/d_{i}$ is a maximally
mixed state of dimension $d_{i}$. Correspondingly, we define the
\textit{single-shot purity cost} as the minimal number of pure single-qubit
states which are required to create the state $\rho$ via unital operations:
\begin{align}
\pur_{\mathrm{c}}^{1}(\rho) & =\min\!\left\{ m:\exists\:\Lambda_{\mathrm{U}}\text{, s.t. }\Lambda_{\mathrm{U}}\!\left[\psi_{2}^{\otimes m}\otimes\frac{\openone}{d_{1}}\right]=\rho\otimes\frac{\openone}{d_{2}}\right\} .\label{eq:puritycost}
\end{align}

Similar quantities were first studied in the asymptotic limit in~\cite{Horodecki2003},
allowing for infinitely many copies of a quantum state and a finite
error margin that only vanishes in this limit. It was found that in
the asymptotic case, the distillable purity and the purity cost coincide,
and are both equal to the relative entropy of purity 
\begin{equation}
\pur_{\mathrm{r}}(\rho)=\log_{2}d-S(\rho).
\end{equation}
The single-copy scenario was considered in \cite{GourResThOpPhysRep15},
under the label of ``nonuniformity''. There the Rényi $\alpha$-purities
were identified as figures of merit using an approach based on Lorentz
curves. We will discuss these results in more detail in the following,
with particular focus on the resource theory of coherence.

\section{\noindent Relation between the resource theories of purity and coherence}

Following established notions from the resource theories of entanglement~\cite{PhysRevLett.78.2275,Bruss2002,Plenio2007,Horodecki2009}
and coherence~\cite{Aberg2006,Plenio2014,Levi2014,WinterResourceCoherence,Marvian2016,CoherenceResource},
we will now introduce a framework for purity quantification. In particular,
we distinguish between \textit{purity monotones} and \textit{purity
measures}. Any purity monotone $\pur$ should fulfill the following
two requirements. 
\begin{lyxlist}{00.00.0000}
\item [{(P1)}] \textit{Nonnegativity}: $\pur$ is nonnegative and vanishes
for the state $\openone/d$. 
\item [{(P2)}] \textit{Monotonicity}: $\pur$ does not increase under unital
operations, i.e., $\pur(\Lambda_{\mathrm{U}}[\rho])\leq\pur(\rho)$
for any unital operation $\Lambda_{\mathrm{U}}$. 
\end{lyxlist}
Similar as in the resource theories of entanglement and coherence,
we regard these two properties as the most fundamental for any quantity
which aims to capture the performance of some purity-based task. Purity
measures will be monotones with the following additional properties. 
\begin{lyxlist}{00.00.0000}
\item [{(P3)}] \textit{Additivity}: $\pur(\rho\otimes\sigma)=\pur(\rho)+\pur(\sigma)$
for any two states $\rho$ and $\sigma$. 
\item [{(P4)}] \textit{Normalization}: $\pur(\ket{\psi}_{d})=\log_{2}d$
for all pure states $\ket{\psi}_{d}$ of dimension $d$. 
\end{lyxlist}
A purity monotone/measure $\pur$ is further convex if it fulfills
$\sum_{i}p_{i}\pur(\rho_{i})\geq\pur(\sum_{i}p_{i}\rho_{i})$. We
note that purity monotones have also been previously studied in~\cite{GourResThOpPhysRep15}.

We can now introduce a family of coherence-based purity monotones
as follows: 
\begin{equation}
\pur_{\coh}(\rho):=\sup_{\Lambda_{\mathrm{U}}}\coh(\Lambda_{\mathrm{U}}[\rho]),\label{eq:PC}
\end{equation}
where the supremum is taken over all unital operations $\Lambda_{\mathrm{U}}$
and $\coh$ is an arbitrary MIO monotone. Clearly, $\pur_{\coh}$
is nonnegative, vanishes for $\openone/d$, and does not increase
under unital operations, i.e., it fulfills the requirements P1 and
P2 for a purity monotone. Remarkably, as we show in Appendix~\ref{sec:Proof-3},
for any MIO monotone $\coh$ the corresponding purity monotone can
be written as 
\begin{equation}
\pur_{\coh}(\rho)=\coh(\rho_{\max})\label{eq:PC-1}
\end{equation}
with the maximally coherent mixed state $\rho_{\max}$. If $\mathcal{C}$
is a distance-based coherence monotone with a contractive distance
$D$, we can apply Theorem~\ref{thm:Cmax} to write the corresponding
purity monotone explicitly as 
\begin{equation}
\pur_{D}(\rho)=D(\rho,\openone/d).\label{eq:PD}
\end{equation}
Eq.~(\ref{eq:PD}) represents a general distance-based purity quantifier,
in direct analogy to similar approaches for entanglement~\cite{PhysRevLett.78.2275,Bruss2002,Plenio2007,Horodecki2009},
coherence~\cite{Plenio2014,CoherenceResource}, and quantum discord~\cite{PhysRevLett.104.080501,RevModPhys.84.1655,Streltsov2014,Adesso2016,Adesso2016b}.
In contrast to these theories, a minimization over free states in
Eq.~(\ref{eq:PD}) is not necessary due to the uniqueness of the
free state in the resource theory of purity. For a single qubit the
relation between coherence and purity can be visualized on the Bloch
ball if coherence and purity are quantified via the trace norm, see
Fig.~\ref{fig:1}. 
\begin{figure}
\centering{}\includegraphics[width=0.58\columnwidth]{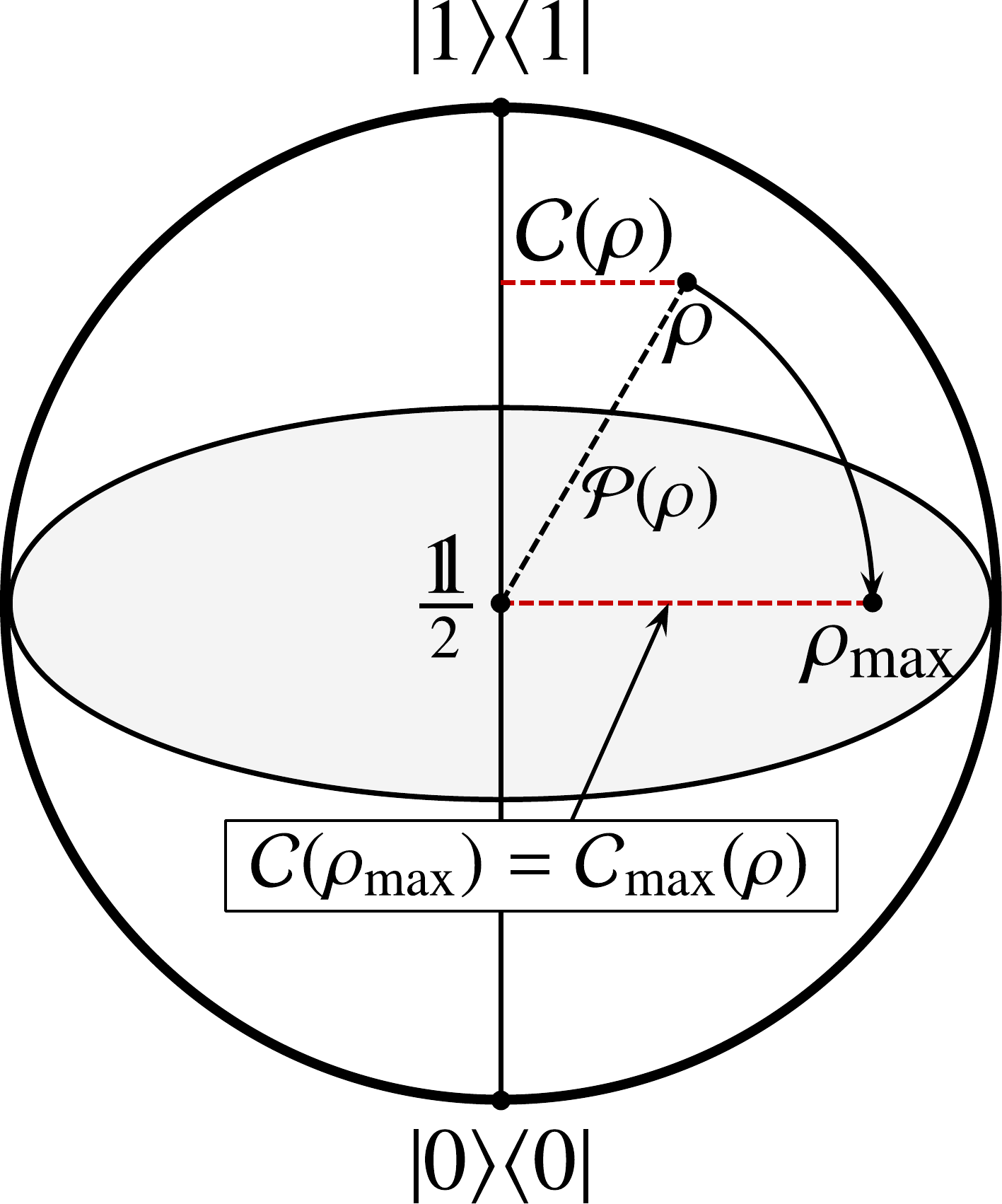}
\caption{\label{fig:1}Coherence and purity for a single qubit. The Bloch ball
of all single-qubit states contains the incoherent axis (line connecting
$\ket{0}\!\bra{0}$ and $\ket{1}\!\bra{1}$) and the maximally coherent
(equatorial) plane. If for a state $\rho$ coherence and purity are
quantified via the distance-based approach with the trace norm $||M||_{1}=\mathrm{Tr}\sqrt{M^{\dagger}M}$,
the corresponding amount of coherence $\coh$ (red dashed lines) and
purity $\pur$ (black dashed line) can be interpreted as the Euclidean
distance to the incoherent axis and the center of the Bloch ball,
respectively. The maximally coherent mixed state $\rho_{\max}$ can
be obtained from $\rho$ via a rotation onto the maximally coherent
plane.}
 
\end{figure}

Cases of particular interest can be derived from the coherence monotones
introduced in Eq.~(\ref{eq:Calpha}). In these cases, Eq.~(\ref{eq:PC})
leads to the \textit{Rényi $\alpha$-purity} 
\begin{align}
\pur_{\alpha}(\rho) & =\log_{2}d-S_{\alpha}(\rho)\label{eq:Palpha}
\end{align}
with the Rényi $\alpha$-entropy $S_{\alpha}(\rho)=\frac{1}{1-\alpha}\log_{2}(\Tr[\rho^{\alpha}])$.
This quantity was studied in \cite{GourResThOpPhysRep15}, and it
admits an operational interpretation in the resource theories of purity.
In particular, the single-shot distillable purity $\pur_{\mathrm{d}}^{1}$,
which was introduced in Eq.~(\ref{eq:distillablepurity}), can be
expressed in terms of $\pur_{\alpha}$ as follows: 
\begin{equation}
\pur_{\mathrm{d}}^{1}(\rho)=\left\lfloor \lim_{\alpha\rightarrow0}\pur_{\alpha}(\rho)\right\rfloor =\left\lfloor \log_{2}(d/r)\right\rfloor ,\label{eq:distillablepurity-1}
\end{equation}
where $r$ is the rank and $d$ is the dimension of $\rho$. Also
the single-shot purity cost $\pur_{\mathrm{c}}^{1}$, introduced in
Eq.~(\ref{eq:puritycost}), can be written in terms of $\pur_{\alpha}$
as follows: 
\begin{equation}
\pur_{\mathrm{c}}^{1}(\rho)=\left\lceil \lim_{\alpha\rightarrow\infty}\pur_{\alpha}(\rho)\right\rceil =\left\lceil \log_{2}(d\lambda_{\max})\right\rceil ,\label{eq:puritycost-1}
\end{equation}
with the maximum eigenvalue $\lambda_{\max}$ of $\rho$. These results
for single-shot purity distillation and dilution were first found
in~\cite{GourResThOpPhysRep15}; we present alternative proofs in
Appendix~\ref{sec:Proof-4} and \ref{sec:Proof-5}. Finally, 
\begin{equation}
\pur_{\mathrm{r}}(\rho)=\lim_{\alpha\rightarrow1}\pur_{\alpha}(\rho)=\log_{2}d-S(\rho)
\end{equation}
 is the \textit{\emph{relative entropy of purity}}. It coincides with
both the distillable purity and the purity cost in the asymptotic
limit, where the resource theory of purity becomes reversible~\cite{Horodecki2003}.

As is summarized in Appendix~\ref{sec:RenyiPurity}, $\pur_{\alpha}$
is a purity measure for all $\alpha\geq0$, i.e., it fulfills all
requirements P1-P4, and it is convex for $0\leq\alpha\leq1$. We further
note that $\pur_{\alpha}(\rho)\geq\pur_{\beta}(\rho)$ for $\alpha\geq\beta$,
since the Rényi entropy $S_{\alpha}$ is nonincreasing in $\alpha$~\cite{Bengtsson2007}.
Aside from the cases discussed before, another case of interest is
the Rényi $2$-purity $\pur_{2}(\rho)=\log_{2}(d\Tr{[\rho^{2}]})$,
a simple function of the linear purity $\Tr[\rho^{2}]$ \footnote{Notice that $\mathrm{Tr}[\rho^{2}]-1/d$ corresponds to the purity
monotone obtained from the squared Schatten 2-norm $||\rho-\openone/d||_{2}^{2}$.}, which can be directly measured by letting two copies of the state
$\rho$ interfere with each other~\cite{Ekert2002,Pichler2013}.
In this way, the purity of a composite system of ultracold bosonic
atoms in an optical lattice, as well as the purity of its subsystems,
have been determined experimentally~\cite{Islam2015}.

\section{Relation to entanglement\protect \\
and quantum discord}

Of particular interest for quantum information theory are non-classical
properties of correlated quantum states in multipartite systems \cite{NielsenChuang,Horodecki2009}.
Our results about purity have immediate consequences for quantities
such as entanglement and discord. Certain relations between entanglement
and purity have already been reported. Bipartite entangled states,
e.g., must have a linear purity above a threshold value of $\Tr[\rho^{2}]=1/(d-1)$,
with total dimension $d$, due to the existence of a finite-volume
set of separable states around the maximally mixed state~\cite{PhysRevA.58.883,PhysRevA.66.062311,Gurvits2003}.
Furthermore, a bound for entanglement can be provided by comparing
the purity of the composite system to the one of its subsystems \cite{Mintert2007,Islam2015}.
Similar investigations have also been performed for multipartite quantum
systems~\cite{GurvitsPhysRevA.72.032322,HildebrandPhysRevA.75.062330}.

In the following we focus on distance-based quantifiers for discord
$\dis$ and entanglement $\ent$, in analogy to Eqs.~(\ref{eq:C})
and~(\ref{eq:PD}). In a multipartite system these can be defined
as~\cite{PhysRevLett.78.2275,PhysRevLett.104.080501} 
\begin{align}
\dis(\rho) & =\inf_{\sigma\in\mathcal{Z}}\dist(\rho,\sigma),\\
\ent(\rho) & =\inf_{\sigma\in\mathcal{S}}\dist(\rho,\sigma),
\end{align}
where $\mathcal{Z}$ and $\mathcal{S}$ denote the sets of zero-discord
and separable states, respectively. The latter contains all convex
combinations of arbitrary product states $\rho_{1}\otimes\cdots\otimes\rho_{N}$,
whereas the set of zero-discord states can either be defined with
respect to a particular subsystem, or symmetrically with respect to
all subsystems. Here, we consider the symmetrical set $\mathcal{Z}$~\cite{RevModPhys.84.1655},
encompassing all convex combinations of pure, locally orthonormal
product states $|\varphi_{1}\rangle\langle\varphi_{1}|\otimes\cdots\otimes|\varphi_{N}\rangle\langle\varphi_{N}|$.
With this choice, $\dis(\rho)$ provides an upper bound for the non-symmetric
definitions of discord.

Our general distance-based approach leads to the following natural
ordering of resources: 
\begin{equation}
\pur(\rho)\geq\coh_{N}(\rho)\geq\dis(\rho)\geq\ent(\rho).
\end{equation}
Here, $\coh_{N}(\rho)=\inf_{\sigma\in\mathcal{I}_{N}}D(\rho,\sigma)$
denotes a coherence monotone with respect to an $N$-partite incoherent
product basis, where $\mathcal{I}_{N}$ is the set of $N$-partite
incoherent states~\cite{Bromley2015,Streltsov2015}. This hierarchy
holds true if purity, coherence, discord, and entanglement are defined
via the same distance $D$. In this case, the statement follows directly
by noting that $\openone/d\in\mathcal{I}_{N}\subset\mathcal{Z}\subset\mathcal{S}$,
see also Fig.~\ref{fig:2}. Based on the same argument, this hierarchy
can be easily extended beyond entanglement to include the concepts
of steering and non-locality~\cite{Wiseman2007,Adesso2016}. It holds
that $\dis(\rho)=\inf_{U_{N}}\coh_{N}(U_{N}\rho U_{N}^{\dag})$ with
product unitary $U_{N}$; this was pointed out for the relative entropy
in~\cite{PhysRevA.92.022112}, but holds for general distance-based
quantifiers. 
\begin{figure}
\includegraphics[width=1\columnwidth]{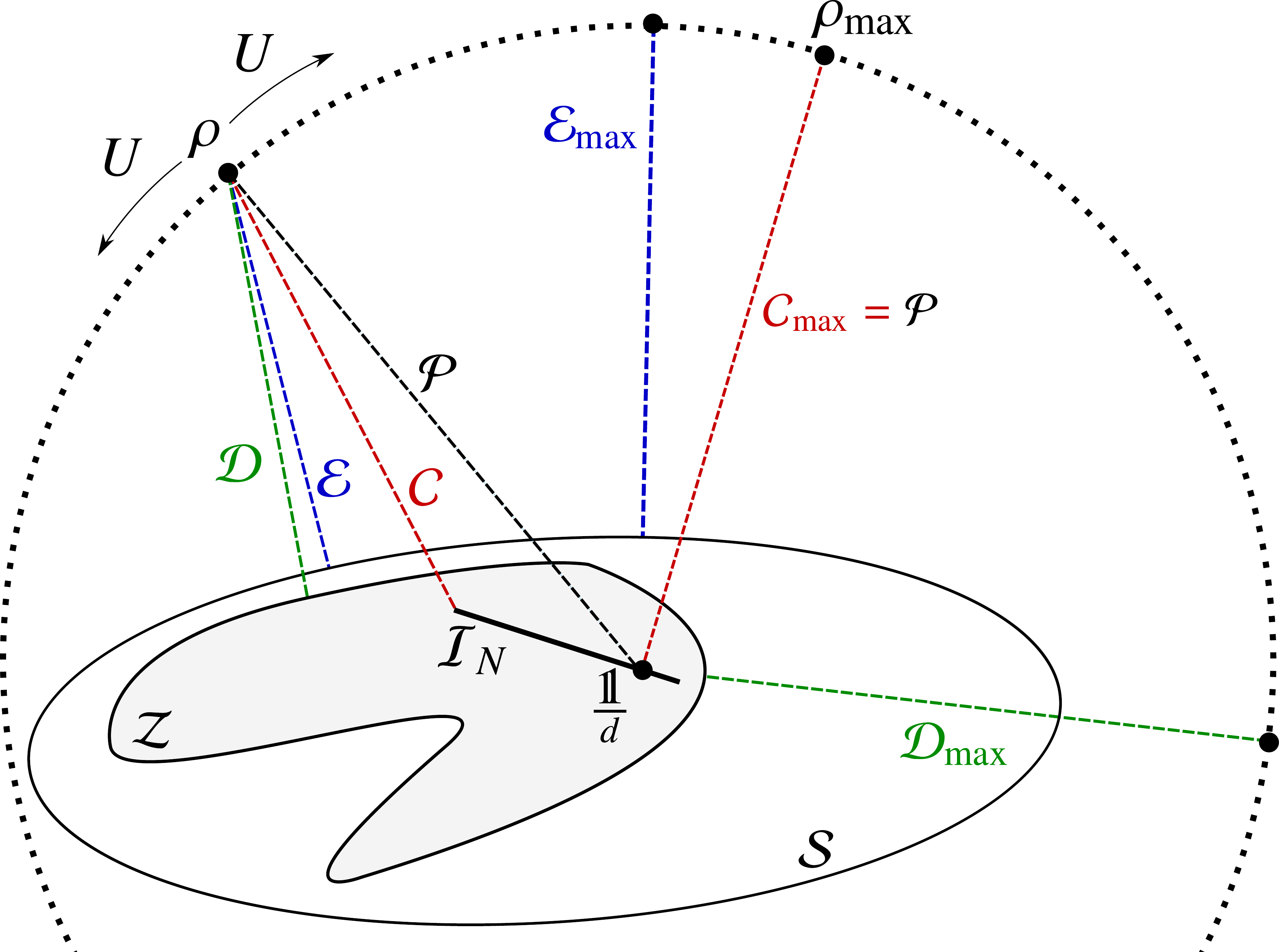}

\caption{\label{fig:2}Schematic representation of purity $\pur$ (black dashed
line), coherence $\coh$ (red dashed lines), discord $\dis$ (green
dashed lines), and entanglement $\ent$ (blue dashed lines) for distance-based
quantifiers of the corresponding framework. Zero-discord states $\mathcal{Z}$
are a nonconvex measure-zero subset of separable states $\mathcal{S}$.
Incoherent states $\mathcal{I}_{N}$ are a convex subset of $\mathcal{Z}$.
All sets contain the maximally mixed state $\openone/d$. The maximally
coherent mixed state $\rho_{\max}$ is obtained from $\rho$ via unitary
rotation (dotted circle).}

\end{figure}

A change of the \textit{global} basis, or equivalently, application
of a collective unitary operation can generate entanglement and discord.
As we will see below, the maximal achievable amount is again directly
bounded by the purity. For this we introduce $\dis_{\max}(\rho)=\sup_{U}\dis(U\rho U^{\dagger})$
and similarly $\ent_{\max}(\rho)=\sup_{U}\ent(U\rho U^{\dagger})$,
in analogy to Eq.~(\ref{eq:Cmax}). As we prove in Appendix~\ref{sec:Proof-6},
these quantities obey the following relation 
\begin{align}
\pur(\rho)=\coh_{\max}(\rho)\geq\dis_{\max}(\rho)\geq\ent_{\max}(\rho),\label{eq:maxhierarchy}
\end{align}
which is visualized in Fig.~\ref{fig:2}. States related to $\dis_{\max}$
and $\ent_{\max}$ have been studied for the two-qubit case. For instance,
the set of maximally entangled mixed states, i.e., states which maximize
entanglement for a fixed spectrum, as well as states that maximize
entanglement at a fixed value of purity, have been characterized for
various quantifiers of entanglement and purity \cite{PhysRevA.62.022310,MaxEntMixed2QubitStatesPRA01,PhysRevA.64.030302,PhysRevA.67.022110}.
States satisfying $\ent_{\max}(\rho)=0$ are also known as absolutely
separable states, and have been studied in~\cite{KusPhysRevA.63.032307,Jivulescu2015276,Filippov1367-2630-19-8-083010}.
Similar studies were performed for discord~\cite{MaxDiscMixed2Qubits,PhysRevA.83.032101},
based on the original definition \cite{Ollivier2002}. We also note
that the relative entropy of purity $\pur_{\mathrm{r}}$ coincides
with the maximal mutual information $I_{\max}(\rho)=\max_{U}I(U\rho U^{\dagger})$,
where $I(\rho)=S(\rho^{A})+S(\rho^{B})-S(\rho)$ is the mutual information
and both subsystems $A$ and $B$ have the same dimension $\sqrt{d}$~\cite{Jevtic2012}.
As a direct consequence of Theorem~\ref{thm:Cmax}, purity further
bounds the accessible entanglement under incoherent operations. This
is discussed in more detail in Appendix~\ref{sec:incoherentent}.

\section{Experimental relevance}

\noindent In well-controllable quantum systems, quantum states with
nearly maximal purity are usually easy to initialize but hard to maintain.
Especially large quantum systems suffer immensely from purity losses
due to noise. For example the linear purity $\Tr[\rho^{2}]$ of the
Greenberger-Horne-Zeilinger (GHZ) state decreases exponentially in
time under global phase noise with a decay proportional to the number
of particles squared~\cite{Monz2011}. A principal challenge of single
photon experiments is the creation of the temporal purity, which is
necessary for coherent interaction between photons of two independent
sources~\cite{Qian2016}.

In contrast to measures of entanglement, discord, or coherence, purity
measures are rather easily accessible in experiments~\cite{Ekert2002,Pichler2013,Islam2015}.
For example the Rényi $\alpha$-purities~(\ref{eq:Palpha}) are essentially
functions of the eigenvalue distribution $\{p_{i}\}$ of $\rho$.
Any von Neumann measurement of $\rho$ immediately provides a lower
bound for this distribution: if such a measurement is performed in
a basis $\{|\varphi_{i}\rangle\}$, the measurement outcomes are distributed
according to the probabilities $p'_{i}=\langle\varphi_{i}|\rho|\varphi_{i}\rangle$.
Any purity monotone that is evaluated on the basis of the $\{p'_{i}\}$
is a lower bound for the true purity of $\rho$ \footnote{This is a direct consequence of the fact that nonselective von Neumann
measurements are unital operations.}. The most easily accessible basis in experiments is the energy eigenbasis.
In this case the measured bound coincides with the purity for all
thermal states. By virtue of Theorem~\ref{thm:Cmax}, the measured
purity also naturally provides an experimental bound on the amount
of coherence, entanglement, and quantum discord.

\section{\noindent Conclusions }

As we have proven in this work, the resource theories of coherence
and purity are closely connected. This connection was established
by showing that any amount of purity can be converted into coherence
by means of a suitable unitary operation. We further provided a closed
expression for the optimal unitary operation, as well as the quantum
states that achieve the maximal coherence. Remarkably, this set of
maximally coherent mixed states is universal, i.e., these states maximize
all coherence monotones for a fixed spectrum. For any distance-based
coherence monotone the maximal coherence achievable via unitary operations
can be evaluated exactly, and is shown to coincide with the corresponding
distance-based purity monotone.

Based on these results, we defined a new family of coherence-based
purity monotones which admit a closed expression and an operational
interpretation in several relevant scenarios. We further proposed
a general framework for quantifying purity, following related approaches
for entanglement and coherence. This approach also provides quantitative
bounds on the required amount of purity to achieve certain levels
of entanglement and discord. Lower bounds for a large variety of purity
measures are easily accessible in experiments.
\begin{acknowledgments}
We thank Gerardo Adesso, Kaifeng Bu, Dario Egloff, Xueyuan Hu, Martin
Plenio, and Alexey Rastegin for discussions. We acknowledge financial
support by the Alexander von Humboldt-Foundation, Bundesministerium
für Bildung und Forschung, the National Science Center in Poland (POLONEZ
UMO-2016/21/P/ST2/04054), and the European Union's Horizon 2020 research
and innovation programme under the Marie Sk\l odowska-Curie grant
agreement No. 665778. 
\end{acknowledgments}

 \bibliographystyle{apsrev4-1}
\bibliography{bibfile}

\appendix

\section{\label{sec:Proof-2}Proof of Theorem \ref{thm:Cmax}}

Here we will prove that  
\begin{equation}
\coh_{\max}(\rho)=\sup_{U}\coh(U\rho U^{\dagger})=\max_{U}\coh(U\rho U^{\dagger})=D(\rho,\openone/d)
\end{equation}
holds true for any distance-based coherence monotone $\coh(\rho)=\inf_{\sigma\in\mathcal{I}}D(\rho,\sigma)$
with a contractive distance $D$, i.e., 
\begin{equation}
D(\Lambda[\rho],\Lambda[\sigma])\leq D(\rho,\sigma)\label{eq:contractivity-2}
\end{equation}
for any quantum operation $\Lambda$.

Our proof will consist of two steps. In the first step, we will prove
the inequality 
\begin{align}
\coh_{\max}(\rho)\leq D\left(\rho,\openone/d\right).\label{eq:bound}
\end{align}
This follows by noting that the maximally mixed state $\openone/d$
is incoherent, and thus gives an upper bound for any distance-based
coherence monotone: $\coh(\rho)\leq D(\rho,\openone/d)$. By contractivity~(\ref{eq:contractivity-2})
the distance $D$ must be invariant under unitaries, which implies
that $\coh(U\rho U^{\dagger})\leq D(\rho,\openone/d)$ for any unitary
$U$. This completes the proof of Eq.~(\ref{eq:bound}).

To complete the proof of the theorem, we will now show the converse
inequality 
\begin{align}
\coh_{\max}(\rho)\geq D\left(\rho,\openone/d\right).
\end{align}
For this, we introduce the unitary $V$ with the property that $\rho_{\max}=V\rho V^{\dagger}$,
where $\rho_{\max}=\sum_{n}p_{n}\ket{n_{+}}\!\bra{n_{+}}$ is a maximally
coherent mixed state. By definition of $\coh_{\max}$, it must be
that $\coh_{\max}(\rho)\geq\coh(V\rho V^{\dagger})$. We further define
$\Delta_{+}$ as the dephasing operation in the maximally coherent
basis: 
\begin{align}
\Delta_{+}[\rho]=\sum_{n}\braket{n_{+}|\rho|n_{+}}\ket{n_{+}}\!\bra{n_{+}}.
\end{align}
It is important to note that the application of $\Delta_{+}$ to any
incoherent state $\sigma\in\mathcal{I}$ leads to the maximally mixed
state: $\Delta_{+}[\sigma]=\openone/d$. If we further define $\tau\in\mathcal{I}$
to be the closest incoherent state to $\rho_{\max}$, we arrive at
the following result: 
\begin{align}
\coh_{\max}(\rho) & \geq\coh(\rho_{\max})=D(\rho_{\max},\tau)\\
 & \geq D\left(\Delta_{+}[\rho_{\max}],\Delta_{+}[\tau]\right)\nonumber \\
 & =D\left(\rho_{\max},\openone/d\right)=D\left(\rho,\openone/d\right).\nonumber 
\end{align}
In the second line we used contractivity~(\ref{eq:contractivity-2})
and in the last equality we used unitary invariance of the distance
$D$. This completes the proof of the theorem.

We note that the same proof also applies for all coherence quantifiers
$\coh_{p}$ based on Schatten $p$-norms for $p\geq1$. This can be
seen by using the same arguments as above, together with the fact
that Schatten $p$-norms are contractive under unital operations for
all $p\geq1$~\cite{Contractivity}.

\section{\label{sec:Proof-3}Proof of Eq.~(\ref{eq:PC-1})}

Here, we will show that any MIO monotone $\coh$ fulfills the following
inequality: 
\begin{equation}
\sup_{\Lambda_{\mathrm{U}}}\coh(\Lambda_{\mathrm{U}}[\rho])=\max_{\Lambda_{\mathrm{U}}}\coh(\Lambda_{\mathrm{U}}[\rho])=\coh(\rho_{\max}),\label{eq:max-unital}
\end{equation}
where $\rho_{\max}=\sum_{n}p_{n}\ket{n_{+}}\!\bra{n_{+}}$ is a maximally
coherent mixed state, $\{p_{n}\}$ is the spectrum of $\rho$, and
the supremum is taken over all unital operations $\Lambda_{\mathrm{U}}$.

In the first step of the proof, we recall that unital operations are
equivalent to mixtures of unitaries with respect to state transformations,
see Lemma 10 in \cite{GourResThOpPhysRep15}. By using similar arguments
as in Appendix~\ref{sec:Proof-2}, we will now show that for any
mixture of unitaries $\Lambda_{\mathrm{MU}}[\rho]=\sum_{i}q_{i}U_{i}\rho U_{i}^{\dagger}$
there exists a maximally incoherent operation $\Lambda_{\mathrm{MIO}}$
such that 
\begin{equation}
\Lambda_{\mathrm{MU}}[\rho_{\max}]=\Lambda_{\mathrm{MIO}}[\rho_{\max}].\label{eq:MU}
\end{equation}
The desired maximally incoherent operation will be given by $\Lambda_{\mathrm{MIO}}[\rho]=\sum_{i,n}K_{i,n}\rho K_{i,n}^{\dagger}$
with Kraus operators $K_{i,n}=\sqrt{q_{i}}U_{i}\ket{n_{+}}\!\bra{n_{+}}$.
It is straightforward to verify that $\sum_{i,n}K_{i,n}\sigma K_{i,n}^{\dagger}=\openone/d$
holds for any incoherent state $\sigma$, which means that the operation
is indeed maximally incoherent. Moreover, it holds that 
\begin{equation}
\sum_{i,n}K_{i,n}\rho_{\max}K_{i,n}^{\dagger}=\sum_{i}q_{i}U_{i}\rho_{\max}U_{i}^{\dagger},
\end{equation}
which completes the proof of Eq.~(\ref{eq:MU}).

Together with Lemma 10 in \cite{GourResThOpPhysRep15}, this result
implies that for any unital operation $\Lambda_{\mathrm{U}}$ there
exists a maximally incoherent operation $\Lambda_{\mathrm{MIO}}$
such that $\Lambda_{\mathrm{U}}[\rho_{\max}]=\Lambda_{\mathrm{MIO}}[\rho_{\max}].$
To complete the proof of Eq.~(\ref{eq:max-unital}), recall that
the states $\rho$ and $\rho_{\max}$ are related via a unitary, i.e.,
$\rho=U\rho_{\max}U^{\dagger}$. This immediately implies that $\coh(\rho_{\max})\leq\sup_{\Lambda_{\mathrm{U}}}\coh(\Lambda_{\mathrm{U}}[\rho])$.
On the other hand, the results presented above imply the converse
inequality: 
\begin{align}
\coh(\rho_{\max}) & \geq\sup_{\Lambda_{\mathrm{U}}}\coh(\Lambda_{\mathrm{U}}[\rho_{\max}])=\sup_{\Lambda_{\mathrm{U}}}\coh(\Lambda_{\mathrm{U}}[\rho]),
\end{align}
where the last equality follows from the fact that $\rho$ and $\rho_{\max}$
are related via a unitary. This completes the proof.

\section{\label{sec:Proof-4}Proof of Eq.~(\ref{eq:distillablepurity-1})}

For proving the statement, let $m$ be an integer such that 
\begin{align}
\Lambda_{\mathrm{U}}\left[\rho\otimes\frac{\openone}{d_{2}}\right]=\psi_{2}^{\otimes m}\otimes\frac{\openone}{d_{1}}\label{eq:transformation-2}
\end{align}
holds true for some unital operation $\Lambda_{\mathrm{U}}$ and some
integers $d_{1}$ and $d_{2}$. Since we require that the unital operation
does not change the dimension of the system, we have the additional
constraint 
\begin{align}
\frac{d_{1}}{d_{2}}=\frac{d}{2^{m}}.\label{eq:condition-3}
\end{align}
From Lemma~\ref{lem:unital}, it follows that the rank of a state
cannot decrease under unital operations. Thus, Eq.~(\ref{eq:transformation-2})
implies 
\begin{align}
\frac{d_{1}}{d_{2}}\geq r,\label{eq:condition-4}
\end{align}
where $r$ is the rank of $\rho$. The inequality~(\ref{eq:condition-4})
implies the majorization relation 
\begin{align}
\rho\otimes\frac{\openone}{d_{2}}\succ\psi^{\otimes m}\otimes\frac{\openone}{d_{1}},
\end{align}
as can be seen by recalling that the maximally mixed state is majorized
by any other state of the same dimension. Thus, by Lemma~\ref{lem:unital},
Eqs.~(\ref{eq:condition-3}) and (\ref{eq:condition-4}) are necessary
and sufficient conditions for the transformation in Eq.~(\ref{eq:transformation-2}).

Eqs.~(\ref{eq:condition-3}) and (\ref{eq:condition-4}) further
imply the inequality 
\begin{align}
m\leq\log_{2}\left(d/r\right),
\end{align}
which proves that single-shot distillable purity is bounded above
by $\left\lfloor \log_{2}\left(d/r\right)\right\rfloor $. Moreover,
it is straightforward to check that Eqs.~(\ref{eq:condition-3})
and (\ref{eq:condition-4}) hold true if we choose $m=\left\lfloor \log_{2}\left(d/r\right)\right\rfloor $,
$d_{1}=d$, and $d_{2}=2^{m}$. This completes the proof.

\section{\label{sec:Proof-5}Proof of Eq.~(\ref{eq:puritycost-1})}

In the first step of the proof, let $m$ be an integer such that $m$
copies of a pure single-qubit state $\psi_{2}$ can be transformed
into the desired state $\rho$ via some unital operation $\Lambda_{\mathrm{U}}$,
i.e., 
\begin{align}
\Lambda_{\mathrm{U}}\left[\psi_{2}^{\otimes m}\otimes\frac{\openone}{d_{1}}\right]=\rho\otimes\frac{\openone}{d_{2}}\label{eq:transformation}
\end{align}
with some integers $d_{1}$ and $d_{2}$. Since we require that $\Lambda_{\mathrm{U}}$
preserves the dimension of the Hilbert space, it must be that 
\begin{align}
\frac{d_{1}}{d_{2}}=\frac{d}{2^{m}}.\label{eq:condition-1}
\end{align}

A necessary requirement for the existence of the unital operation
in Eq.~(\ref{eq:transformation}) is that due to Lemma~\ref{lem:unital}
the maximal eigenvalue of $\rho\otimes\openone/d_{2}$ -- which is
$\lambda_{\max}/d_{2}$ -- is smaller or equal than the maximal eigenvalue
of the resource state $\psi_{2}^{\otimes m}\otimes\openone/d_{1}$,
i.e. 
\begin{align}
\frac{\lambda_{\max}}{d_{2}}\leq\frac{1}{d_{1}}.\label{eq:condition-2}
\end{align}
It is now crucial to note that due to the special form of the resource
state, Eq.~(\ref{eq:condition-2}) directly implies the majorization
relation 
\begin{align}
\rho\otimes\frac{\openone}{d_{2}}\prec\psi_{2}^{\otimes m}\otimes\frac{\openone}{d_{1}}.
\end{align}
Thus, by Lemma~\ref{lem:unital}, Eqs.~(\ref{eq:condition-1}) and
(\ref{eq:condition-2}) are necessary and sufficient for the transformation
in Eq.~(\ref{eq:transformation}).

In the next step, we note that Eqs.~(\ref{eq:condition-1}) and (\ref{eq:condition-2})
imply the following inequality: 
\begin{align}
m\geq\log_{2}(d\lambda_{\max}),
\end{align}
which means that the single-shot purity cost is bounded below by $\left\lceil \log_{2}(d\lambda_{\max})\right\rceil $.
In the last step, it is straightforward to check that Eqs.~(\ref{eq:condition-1})
and (\ref{eq:condition-2}) hold true if we choose $m=\left\lceil \log_{2}(d\lambda_{\max})\right\rceil $,
$d_{1}=d$, and $d_{2}=2^{m}$. This completes the proof.

\section{\label{sec:RenyiPurity}Properties of Rényi $\alpha$-purities }

Here we will prove that the Rényi $\alpha$-purity 
\begin{align}
\pur_{\alpha}(\rho)=\log_{2}d-S_{\alpha}(\rho)
\end{align}
is a purity measure, i.e., it fulfills the requirements P1-P4 stated
in the main text. For this, we will use the fact that the Rényi entropy
is Schur concave for all $\alpha\geq0$~\cite{Olkin2011}: 
\begin{align}
\rho\succ\sigma\Rightarrow S_{\alpha}(\rho)\leq S_{\alpha}(\sigma).\label{eq:SchurConcavity}
\end{align}
We will now prove each of the conditions P1-P4. 
\begin{itemize}
\item[P1] $\pur_{\alpha}(\openone/d)=0$ follows immediately from $S_{\alpha}(\openone/d)=\log_{2}d$
for all $\alpha$. Furthermore Eq.~(\ref{eq:SchurConcavity}) and
the fact that the maximally mixed state $\openone/d$ is majorized
by any other state of the same dimension imply nonnegativity: 
\begin{align}
\pur_{\alpha}(\rho)\geq\pur_{\alpha}(\openone/d)=0.
\end{align}

\item[P2] Due to Lemma~\ref{lem:unital}, we have $\rho\succ\Lambda_{\mathrm{U}}[\rho]$
for any unital operation $\Lambda_{\mathrm{U}}$. Eq.~(\ref{eq:SchurConcavity})
then implies that $\pur_{\alpha}(\rho)\geq\pur_{\alpha}(\Lambda_{\mathrm{U}}[\rho])$. 
\item[P3] The Rényi entropy is additive: $S_{\alpha}(\rho\otimes\sigma)=S_{\alpha}(\rho)+S_{\alpha}(\sigma)$.
This directly implies additivity of $\pur_{\alpha}$. 
\item[P4] $\pur_{\alpha}(\ket{\psi}_{d})=\log_{2}d$, since $S_{\alpha}(\ket{\psi}_{d})=0$
for all $\alpha$. 
\end{itemize}
The Rényi $\alpha$-purity is convex for $0\leq\alpha\le1$, since
$S_{\alpha}$ is concave in this region~\cite{Bengtsson2007}. For
$\alpha>1$ the Rényi entropy $S_{\alpha}$ is neither concave nor
convex \cite{RenyiEntrConcFiniteAlp78}.

\section{\label{sec:Proof-6}Proof of Eq.~(\ref{eq:maxhierarchy}) }

Let us denote with $U_{\ent}$ the unitary operation that provides
the maximum for $\ent_{\max}(\rho)$. We find 
\begin{align}
\ent_{\max}(\rho) & =\ent(U_{\ent}\rho U_{\ent}^{\dagger})\leq\dis(U_{\ent}\rho U_{\ent}^{\dagger})\notag\\
 & \leq\sup_{U}\dis(U\rho U^{\dagger})=\dis_{\max}(\rho).
\end{align}
Similarly, let $U_{\dis}$ be the unitary that leads to $\dis_{\max}(\rho)$.
We obtain 
\begin{align}
\dis_{\max}(\rho) & =\dis(U_{\dis}\rho U_{\dis}^{\dagger})\leq\coh_{N}(U_{\dis}\rho U_{\dis}^{\dagger})\leq\sup_{U}\coh_{N}(U\rho U^{\dagger})\notag\\
 & =\sup_{U}\coh(U\rho U^{\dagger})=\pur(\rho),
\end{align}
where we used Theorem~\ref{thm:Cmax} as well as the fact that any
two bases can be mapped onto each other by a unitary operation.

\section{\label{sec:incoherentent}Purity bounds on entanglement\protect
\protect \protect \protect \\
 by incoherent operations}

The amount of entanglement which can be generated by an optimal incoherent
operation is bounded by the coherence~\cite{Streltsov2015}: 
\begin{align}
\coh_{\mathrm{r}}(\rho^{A})= & \lim_{d_{B}\rightarrow\infty}\left\{ \sup_{\Lambda_{\mathrm{i}}}\ent_{\mathrm{r}}^{A:B}\left(\Lambda_{\mathrm{i}}\left[\rho^{A}\otimes\proj{0}^{B}\right]\right)\right\} ,\label{eq:activation}
\end{align}
where the supremum is performed over all bipartite incoherent operations
$\Lambda_{\mathrm{i}}$~\cite{Streltsov2015} and $\coh_{\mathrm{r}}$
and $\ent_{\mathrm{r}}$ are the relative entropy of coherence and
entanglement respectively. Our results from Theorem~\ref{thm:Cmax}
allow us to further connect these results to the relative entropy
of purity: Using a unitary to rotate $\rho^{A}$ into a maximally
coherent basis followed by the application of the optimal incoherent
operation, the generated entanglement amounts to 
\begin{align}
\pur_{\mathrm{r}}(\rho^{A})= & \sup_{U}\lim_{d_{B}\rightarrow\infty}\left\{ \sup_{\Lambda_{\mathrm{i}}}\ent_{\mathrm{r}}^{A:B}\left(\Lambda_{\mathrm{i}}\left[U\rho^{A}U^{\dag}\otimes\proj{0}^{B}\right]\right)\right\} 
\end{align}
with the relative entropy of purity $\pur_{\mathrm{r}}$.

A similar result can be established for the geometric entanglement
$\ent_{\mathrm{g}}(\rho)=1-\max_{\sigma\in\mathcal{S}}F(\rho,\sigma)$
and the geometric coherence $\coh_{\mathrm{g}}(\rho)=1-\max_{\sigma\in\mathcal{I}}F(\rho,\sigma)$,
recalling that Eq.~(\ref{eq:activation}) also holds true for these
quantities~\cite{Streltsov2015}. If we introduce the geometric purity
as $\pur_{\mathrm{g}}(\rho)=1-F(\rho,\openone/d)=1-\frac{1}{d}(\Tr\sqrt{\rho})^{2}$,
we immediately obtain the following result: 
\begin{align}
\pur_{\mathrm{g}}(\rho^{A})= & \sup_{U}\lim_{d_{B}\rightarrow\infty}\left\{ \sup_{\Lambda_{\mathrm{i}}}\ent_{\mathrm{g}}^{A:B}\left(\Lambda_{\mathrm{i}}\left[U\rho^{A}U^{\dag}\otimes\proj{0}^{B}\right]\right)\right\} .
\end{align}

In \cite{Orieux2015} a CNOT-gate ($U_{\mathrm{CNOT}}$) is used to
create entanglement out of the two-qubit input state 
\begin{align}
\rho_{\text{in}}=\rho^{A}\otimes\proj{0}^{B}
\end{align}
with system $A$ being the control qubit system and $B$ being the
target qubit, i.e.\ $\rho_{\text{out}}=U_{\mathrm{CNOT}}\rho_{\text{in}}U_{\mathrm{CNOT}}^{\dag}$.
In this two-qubit scenario the entanglement of the state $\rho_{\text{out}}$
can be measured by the negativity $\mathcal{N}(\rho)=\sum_{j}|\lambda_{j}^{-}|$
where $\lambda_{j}^{-}$ are the negative eigenvalues of the partial
transpose of $\rho$~\cite{Peres1996,Horodecki1996,PhysRevA.58.883,Vidal2002}.
The negativity of $\rho_{\mathrm{out}}$ is closely related to the
$l_{1}$-norm of coherence of the state $\rho^{A}$~\cite{Nakano2013}:
\begin{align}
\mathcal{N}(\rho_{\text{out}})=\abs{\rho_{01}^{A}}=\frac{\coh_{\ell_{1}}(\rho^{A})}{2},
\end{align}
with $\rho_{01}^{A}$ being the off-diagonal element of the chosen
qubit basis.

For a single qubit there is a direct relation between $\coh_{\ell_{1}}$
and the geometric coherence \cite{Streltsov2015}: $C_{\ell_{1}}=\sqrt{1-(1-2C_{\mathrm{g}})^{2}}$.
Using Theorem~\ref{thm:Cmax} to bound the geometric coherence by
the geometric purity, we obtain the following bound for the negativity
\begin{align}
\mathcal{N}(\rho_{\text{out}})\leq\sqrt{1-(1-2\pur_{\mathrm{g}})^{2}},
\end{align}
where equality holds if the eigenstates of $\rho^{A}$ form a maximally
coherent basis.
\end{document}